\documentclass[11pt]{article}
\usepackage{bbm,amsmath,amssymb,url,graphicx,mathtools,amsthm,color,booktabs}
\usepackage{fullpage}
\usepackage{listings}
\usepackage{mathrsfs}
\usepackage{epsfig,graphicx,caption,subcaption, color}
\usepackage[colorlinks,pdftex]{hyperref} 
\usepackage[margin=0.8in]{geometry}

\DeclarePairedDelimiter\abs{\lvert}{\rvert}
\newcommand{\bra}[1]{\left\langle #1\right|}
\newcommand{\ket}[1]{\left|#1\right\rangle}

\newcommand{\braket}[2]{\left\langle #1\middle|#2\right\rangle}

\newtheorem{proposition}{Proposition}
\newtheorem{theorem}{Theorem}
\newtheorem{lemma}{Lemma}

\def\be#1\ee{\begin{equation}#1\end{equation}}
\def\ba#1\ea{\begin{align}#1\end{align}}

\begin{document}
\title{\Large Quantum ground state isoperimetric inequalities for the\\ energy spectrum of local Hamiltonians}
\author{ Elizabeth Crosson\footnote{University of New Mexico. {\tt crosson@unm.edu}}, John Bowen\footnote{University of Chicago. {\tt johnbowen@uchicago.edu}}}

\maketitle

\abstract{We investigate the relationship between the energy spectrum of a local Hamiltonian and the geometric properties of its ground state.  By generalizing a standard framework from the analysis of Markov chains to arbitrary (non-stoquastic) Hamiltonians we are naturally led to see that the spectral gap can always be upper bounded by an isoperimetric ratio that depends only on the ground state probability distribution and the range of the terms in the Hamiltonian, but not on any other details of the interaction couplings.   This means that for a given probability distribution the inequality constrains the spectral gap of any local Hamiltonian with this distribution as its ground state probability distribution in some basis (Eldar and Harrow derived a similar result~\cite{eldar2015local} in order to characterize the output of low-depth quantum circuits).  Going further, we relate the Hilbert space localization properties of the ground state to higher energy eigenvalues by showing that the presence of $k$ strongly localized ground state modes (i.e. clusters of probability, or subsets with small expansion) in Hilbert space implies the presence of $k$ energy eigenvalues that are close to the ground state energy.    Our results suggest that quantum adiabatic optimization using local Hamiltonians will inevitably encounter small spectral gaps when attempting to prepare ground states corresponding to multi-modal probability distributions with strongly localized modes, and this problem cannot necessarily be alleviated with the inclusion of non-stoquastic couplings.


\section{Introduction}
Isoperimetric inequalities apply to mathematical spaces in which one has both a notion of measure and a notion of locality, so that the measure inside a subset can be compared with the measure along its boundary.  The isoperimetric constant of a space is the minimum isoperimetric ratio of subsets with less than half the total measure,
$$
\textrm{isoperimetric constant} = \min_{\textrm{subsets}\;\subset\;\textrm{space}} \frac{\textrm{measure on the boundary of the subset}}{\textrm{measure inside the subset}}.
$$

Perhaps the most well known example of an isoperimetric inequality is the statement that the area $A$ enclosed by a curve in the plane of length $L$ satisfies $L^2 \geq 4 \pi A$, with equality holding if and only if the curve is a circle.   As a discrete example, if $S \subset \mathbb{Z}^d$ is a finite subset and $E(S,\bar{S})$ is the set of adjacent pairs of vertices $(v,w)$ with $v \in S$ and $w \in \bar{S} := \mathbb{Z}^d - S$, then the cardinality of these sets satisfy $|E(S,\bar{S})| \geq 2 d |S|^{\frac{d-1}{d}}$. 

In graph theory and in the analysis of Markov chains there is a particularly fruitful connection between the isoperimetric constant of a graph and the spectral properties of a matrix associated with the graph (e.g. the adjacency matrix, the graph Laplacian, or the Markov chain transition matrix).  This connection has had a significant impact on theoretical computer science, leading to the notion of spectral expander graphs~\cite{alon1984eigenvalues} and later to geometric methods for bounding the convergence time of Markov chain Monte Carlo algorithms~\cite{aldous1987markov, lawler1988bounds, sinclair1989approximate}.  In all of these settings an isoperimetric constant of an extremal eigenvector (corresponding to the maximum or minimum eigenvalue) is related to the next largest (or smallest) eigenvalue in the spectrum.  

In this work our goal is to generalize some aspects of this connection between algebra and geometry to the setting of quantum Hamiltonians.   Here we consider any probability distribution over measurement outcomes as a a weighted graph, with unweighted edges induced by the locality of the Hamiltonian, and use isoperimetric ratios of this graph to upper bound the excited energy eigenvalues of any local Hamiltonian with that distribution as its ground state probability distribution in some basis.  

In addition to the broad interest in rigorously provable properties of ground states of local Hamiltonians, a primary motivation is to further our understanding of adiabatic quantum computation.  In the idealized version of the adiabatic model the computational state of the system is the instantaneous ground state of a local Hamiltonian that varies slowly as a parameter $s$ goes from 0 to 1.  The system is initially prepared in a product state ground state at $s = 0$, and the Hamiltonian $H(s)$ varies continuously until reaching the final Hamiltonian at $s = 1$, which is designed so that samples from this ground state probability distribution in a particular basis can be used to solve computational problems such as classical optimization~\cite{farhi2000quantum} or simulating the quantum circuit model with polynomial overhead~\cite{aharonov2008adiabatic}.  The run time of an adiabatic computation scales polynomially with the system size and the inverse of the minimum spectral gap above the ground state during the evolution.    

Constraining the spectral gap in terms of the geometry of the ground state probability distribution reveals some of the ultimate limits within the strict model of idealized adiabatic computation (some probability distributions, even quite simple ones, will inevitably take exponential time to be precisely reached by a purely adiabatic computation), while also showing the necessity of removing bottlenecks in the ground state geometry to improve the performance within the adiabatic paradigm.  Of particular interest is the possibility of improving the performance of adiabatic optimization using non-stoquastic driver terms~\cite{farhi2002quantum,susa2016relation,nishimori2016exponential,hormozi2016non,crosson2014different} as an alternative to traditional transverse field annealing.   Although ground state wave functions of non-stoquastic Hamiltonians have amplitudes with arbitrary complex phases, the absolute square of the wave function still leads to a graph with non-negative vertex weights, and we show that the geometry of this graph constrains the low-lying excited energies of the Hamiltonian in the same way as it does for stoquastic Hamiltonians and for Markov chains. 

Now we will formally state some of our results before outlining past related work and the organization of the later sections.  

\paragraph{Formal setting.}
Let $H = \sum_{i = 1}^m H_i$ be a local Hamiltonian with operator norm $\|H\|$ acting on a finite dimensional Hilbert space of dimension $N + 1$ with non-negative energies $E := E_0 \leq E_1 \leq ... \leq E_N$ corresponding to eigenstates $|\psi\rangle := |\psi_0\rangle , |\psi_1\rangle , ... , |\psi_N\rangle$.  Define the spectral gap $\Delta_H := E_1 - E_0$.  Choose any complete set of basis labels $\mathcal{B}$ and define $\psi_z:= \langle z | \psi\rangle$ for $z\in \mathcal{B}$.  
For any subset $S\subseteq \mathcal{B}$ define $\pi(S) := \sum_{z\in S}\psi^*_z\psi_z$, and let $\partial S := \{x\in S : \exists y\in \mathcal{B} - S\; \text{with} \; \bra{x} H \ket{y} \neq 0\}$ denote the interior boundary of $S$.  
\begin{theorem}\label{theo:main}
If $S \subset \mathcal{B}$ with $0 < \pi(S) \leq 1/2$ then the spectral gap $\Delta_H$ satisfies
\be
\Delta_H \leq 2 (\|H\| - E) \frac{\pi(\partial S)}{\pi(S)}. \label{eq:expansion}
\ee
\end{theorem}

The generalization of Theorem \ref{theo:main} to the higher eigenvalues of $H$ uses the following definition: a collection of pairwise disjoint subsets $S_0 ,... ,S_k \subset \mathcal{B}$ are said to be isolated if $\langle x | H | y\rangle = 0$ whenever $x \in S_i$ and $y\in S_j$ for $i \neq j$.  Simply put, two subsets of basis states are isolated if the off-diagonal elements of the Hamiltonian do not connect them.  

\begin{theorem}\label{thm:2}
If $S_0, ... , S_k \subset \mathcal{B}$ are isolated subsets with $0 < \pi(S_i) \leq 1/2$ for all $i$ then
\begin{equation}
\frac{E_k - E}{2(\|H\| - E)} \leq  \max_{i=0,...,k} \frac{\pi(\partial S_i)}{\pi(S_i)}.\label{eq:multiexpansion}
\end{equation}

\end{theorem}
According to \eqref{eq:multiexpansion}, if the ground state probability distribution of the system can be split into $k +1$ clusters which are each localized in Hilbert space, then the lowest $k$ excited energies will be close to the ground state energy.  

In addition to theorems \ref{theo:main} and \ref{thm:2} above we also give two new results related to lower bounds on the spectral gap.  In section \ref{sec:stoquastic} a vertex expansion spectral gap lower bound for irreducible stoquastic Hamiltonians which uses assumptions which are natural in the setting of adiabatic optimization, and in section \ref{sec:counterexample} we provide a simple counterexample to the possibility of directly generalizing Cheeger's inequality to nonstoquastic Hamiltonians.  

\paragraph{Previous work.}  If $H$ and $\mathcal{B}$ are chosen so that $\langle z | H | z'\rangle \leq 0$ for all $z,z' \in \mathcal{B}$ then the Hamiltonian is called stoquastic~\cite{bravyi2006complexity}.   Stoquastic Hamiltonians can be rescaled and shifted into non-negative matrices, and by a further quantum-to-classical mapping~\cite{bravyi2009complexity} the bound \eqref{eq:expansion} reduces to a previously known bound in terms of a quantity called the Markov chain conductance.\footnote{Markov chain conductance includes weights on the edges, but since these are positive and no greater than 1, the conductance upper bound can always be relaxed to the vertex expansion form that we present.}  This observation about stoquastic Hamiltonians has appeared previously in the quantum information literature~\cite{al2010energy, jarret2015adiabatic}.  The multi-expansion bound \eqref{eq:multiexpansion} is also known for Markov chains and graph Laplacians, and in those cases it can be applied when $S_0, ... ,S_k$ are disjoint but not necessarily isolated, however to our knowledge it has not been previously stated for any class of Hamiltonians.  

For Markov chains and graphs with positive weights on the edges and vertices, there is a celebrated result called Cheeger's inequality that provides a lower bound on the spectral gap in terms of the smallest isoperimetric ratio over all subsets in the graph which are not too large.  The generalization of Cheeger's inequality for non-negative matrices first appeared in~\cite{lawler1988bounds}, and has also been stated for stoquastic Hamiltonians~\cite{al2010energy, jarret2015adiabatic}.  In section \ref{sec:stoquastic} we prove the stoquastic Cheeger inequality by explicitly mapping the Hamiltonian to a Markov chain (which differs only slightly from past approaches~\cite{al2010energy, jarret2015adiabatic}).  We also derive a new vertex expansion form of Cheeger's inequality for stoquastic Hamiltonians that correspond to irreducible matrices.  Finally, in section \ref{sec:counterexample} we show using a counterexample that a direct generalization of Cheeger's inequality fails to hold for non-stoquastic Hamiltonians.  

The observation \eqref{eq:expansion} that the spectral gap of the Hamiltonian lower bounds the ground state vertex expansion of subsets of basis states also appeared recently in the work of Eldar and Harrow~\cite{eldar2015local}.  Their goal was to show that the output of log-depth quantum circuits is always highly expanding in (at least one of) the $Z$ or $X$ bases, and they accomplished this by relating the circuit output distribution to the ground state of a gapped Hamiltonian which has an interaction range scaling with the depth of the circuit, and then proving a result that is equivalent to \eqref{eq:expansion}.  Our contribution in this work is to derive this result through a framework that connects non-stoquastic Hamiltonians with some aspects of the theory of Markov chains, and to emphasize the relevance of this result for understanding the potential advantages and limitations of quantum adiabatic optimization with non-stoquastic Hamiltonians.  Our framework also leads to additional results not present in \cite{eldar2015local} such as the relation between clustering and higher eigenvalues \eqref{eq:multiexpansion}, as well as generalizations of Dirichlet form and conductance upper bounds in section \ref{Proofs}, while also raising some new open questions and future directions.

\paragraph{Organization of the remaining sections.} In the next section we establish upper bounds which logically precede \eqref{eq:expansion}, which involve a quantity we call the Hamiltonian conductance because it generalizes the concept of edge expansion in graphs or conductance in Markov chains.  These results are obtained from the properties of an operator which formally generalizes a reversible Markov chain transition matrix, which we term an approximate ground state probability projector.   After deriving the main theorem from this conceptual framework we give a more direct proof using the variational principle in section \ref{sec:variational}, which has the advantage of being straightforward to verify but does not provide intuition for the steps along the way.  In section \ref{sec:examples} we discuss some example applications of theorems \ref{theo:main} and \ref{thm:2}.  In section \ref{sec:stoquastic} we consider stoquastic Hamiltonians and show that the spectral gap can be lower bounded in terms of vertex expansion for a class of models including the generalized transverse Ising models.  In section \ref{sec:counterexample} we rule out the possibility of a direct generalization of Cheeger's inequality for non-stoquastic Hamiltonians by using an explicit counterexample.  Finally in section \ref{sec:adiabatic} we discuss some implications of these results for adiabatic quantum computation.

\section{Proofs}\label{Proofs}
Throughout this section we will consider a Hamiltonian $H$ with eigenvalues and eigenstates as defined in the introduction.  In principle we let $\mathcal{B}$ be (a set of labels for) an arbitrary basis of the Hilbert space on which $H$ acts, but in general it is useful to keep in mind the computational basis $\mathcal{B} = \{0,1\}^n$ or other tensor product bases to take advantage of the locality of $H$.

\subsection{\normalsize Projecting onto the ground state probability distribution}
Define $G$ to be a shifted and rescaled version of $H$,
\begin{equation}
G := \frac{\|H\| - H}{\|H\| - E},\label{eq:G}
\end{equation}
and note that the ground state $|\psi\rangle$ of $H$ becomes the principal eigenvector of $G$ with eigenvalue 1.  Let $\Omega$ be the support of $|\psi\rangle$ in the basis $\mathcal{B}$,
\be
\Omega := \{z \in \mathcal{B} :  \langle z | \psi\rangle \neq 0\}.
\ee 
For any $S \subseteq \mathcal{B}$ define the projector $\mathbbm{1}_S := \sum_{z \in \mathcal{S}} |z\rangle \langle z|$.  Let $D := \sum_{z\in \Omega} \langle z | \psi\rangle | z\rangle \langle z|$ and define
\begin{equation}
P := D^{-1} \mathbbm{1}_\Omega G \mathbbm{1}_\Omega D = \sum_{x,y\in \Omega} \frac{\langle y | \psi\rangle}{\langle x | \psi \rangle} \langle x | G | y\rangle |x\rangle \langle y| .\label{eq:pee}
\end{equation}
In the next section we will explore the properties of $P$ and derive bounds on its spectrum of eigenvalues $p_0 \geq p_1 \geq ... \geq p_{|\Omega|-1} \geq 0$.  To relate these eigenvalues back to those of $H$, note that $P$ and $\mathbbm{1}_\Omega G \mathbbm{1}_\Omega$ are related by a similarity transformation and so they have the same eigenvalues, and the Cauchy interlacing theorem implies that the largest $|\Omega|$ eigenvalues of $\mathbbm{1}_\Omega G \mathbbm{1}_\Omega$ are no smaller than the corresponding eigenvalues of $G$, and so
\be
\frac{E_k - E}{\|H\| - E} \leq p_0 - p_k \label{eq:fullspectrumP}.
\ee
The fact that $|\psi\rangle$ is an eigenstate of $\mathbbm{1}_\Omega G \mathbbm{1}_\Omega$ with eigenvalue 1 implies that $p_0 = 1$ is the eigenvalue corresponding to the left eigenvector $\langle \pi | := \langle \psi | D = \sum_{z \in \Omega} |\psi_z|^2 \langle z|$, and so the spectral gap $\Delta_P := 1 - p_1$ is related to the spectral gap $\Delta_H$ by
\begin{equation}
\frac{\Delta_H}{\|H\| - E} \leq \Delta_P .\label{eq:gap}
\end{equation}
Both \eqref{eq:fullspectrumP} and \eqref{eq:gap} are equalities if $\Omega = \mathcal{B}$ i.e. if the ground state wave function has support everywhere in the basis $\mathcal{B}$.  

The operator $P$ obtained from a similarity transformation of the Hamiltonian has been used extensively in the quantum information literature to relate stoquastic Hamiltonians to Markov chains~\cite{bravyi2006merlin, aharonov2008adiabatic, bravyi2009complexity,somma2013spectral, bravyi2014monte, movassagh2016supercritical, bausch2016increasing, levine2016gap}.   This can be seen by defining $P_{x,y} := \langle x | P |y\rangle$ and observing that $\sum_{y\in\Omega} P_{x,y} = 1$ for all $x \in \Omega$,
\begin{equation}
\sum_{y \in \Omega} P_{x,y} = \sum_{y \in \Omega} \frac{\langle y | \psi\rangle}{\langle x | \psi \rangle} \langle x | G| y\rangle = \frac{\langle x| G  \mathbbm{1}_{\Omega} | \psi \rangle}{\langle x | \psi \rangle}= \frac{\langle x| G  | \psi \rangle}{\langle x | \psi \rangle} = 1. \label{eq:rowstochastic}
\end{equation}
If $H$ is stoquastic in the basis $\mathcal{B}$, then it turns out that $P_{x,y} \geq 0$ for all $x,y\in \Omega$ (see section \ref{sec:stoquastic}) and so the property \eqref{eq:rowstochastic} implies that $P$ is a Markov chain generator (we will explore this special setting further in section \ref{sec:stoquastic}).  

When $H$ is non-stoquastic in the basis $\mathcal{B}$ the operator $P$ still resembles a reversible Markov chain transition matrix in the sense of the following properties:
\begin{enumerate}
\item $\sum_{y \in \Omega} P_{x,y} = 1$ for all $x\in \Omega$. 
\item $\sum_{x\in \Omega} \pi_x P_{x,y} = \pi_y$. 
\item $\pi_x P_{x,y} = \pi_y P_{y,x}^*$.  
\end{enumerate}
The first property says that the rows of the matrix form of $P$ all sum to 1, the second property reiterates the fact that $\langle \pi | P = \langle \pi|$, and the third property says that the ``stationary distribution'' $\pi$ and the ``transition quasi-probabilities'' $P_{x,y}$ satisfy a complex generalization of detailed balance.  Although $P$ is not a stochastic matrix when $H$ is non-stoquastic, it still preserves the ground state probability distribution while decreasing the norm of the subspace orthogonal to $\langle \pi |$ by a factor $1 - \Delta_P$.  Therefore we may regard $P$ as an approximate projector onto the ground state probability distribution of $H$, even when $H$ is non-stoquastic.  
\subsection{Analyzing $P$ in a $\pi$-weighted Hilbert space}

While the matrix $P$ is typically non-Hermitian, this section generalizes a standard Markov chain approach~\cite{levin2009markov} to define a weighted complex inner product space in which $P$ is always self-adjoint.  This new inner product space has the same underlying vector space as the standard quantum formalism, so to facilitate a common notation we will think of this as a vector space of functions $\mathbb{C}^\Omega:= \{f:\Omega \rightarrow \mathbb{C}\}$.  The usual quantum Hilbert space is $L^2(\Omega)$, with the complex inner product of vectors $f,g \in \mathbb{C}^\Omega$ represented by bracket notation, $\langle f | g \rangle := \sum_{x\in \Omega} f^*_x g_x$ . The alternate Hilbert space considered in this section is $L^2(\Omega, \pi)$ with a $\pi$-weighted complex inner product,
\begin{equation}
\langle f, g \rangle_\pi := \sum_{x\in \Omega}\pi_x f^*_x g_x.
\end{equation}
The fact that $P$ is a self-adjoint operator with respect to the $\langle \cdot , \cdot \rangle_\pi$ complex inner product follows because of the complex detailed balance property $\pi_x P_{x,y} = \pi_y P^*_{y,x}$,
\begin{align}
\langle f, P g\rangle_\pi = \sum_{x,y\in \Omega} \pi_x f^*_x P_{x,y} g_y 
& =\sum_{y\in \Omega} \pi_y \left( \sum_{x\in \Omega} f_x  P_{y,x}\right)^* g_y = \langle P f, g\rangle_\pi.
\end{align}
The constant function $1_\Omega: \Omega \rightarrow {1}$ is the eigenvector of $P$ with eigenvalue $1$.  A function $f$ that is orthogonal to the constant function has vanishing expectation when it is viewed as a complex variable with distribution $\pi$,  $\langle 1, f\rangle_\pi  = \sum_{x\in \Omega} \pi_x f_x := \mathbf{E}_\pi(f) = 0$. 

If $\{|\psi^{\Omega,k}\rangle: k = 0 ,..., |\Omega| - 1 \}$ is the orthornomal eigenbasis of $\mathbbm{1}_\Omega G \mathbbm{1}_{\Omega}$, then the corresponding vectors $\{f^k\}_{k = 0}^{|\Omega|-1}$ after applying the similarity transformation will be orthonormal with respect to the $\langle \cdot , \cdot \rangle_\pi$ inner product,  
\begin{align*}
\langle f^k , f^{k'}\rangle_\pi &= \sum_{x\in \Omega} \pi_x \left(\frac{\psi_x^{\Omega,k *}}{\psi^*_x}\right)\left(\frac{\psi_x^{\Omega,k'}}{\psi_x}\right)\\
&= \sum_{x\in \Omega} \psi^{\Omega, k *}_x \psi^{\Omega, k'}_x \\
&= \langle \psi^{\Omega, k} | \psi^{\Omega, k'}\rangle = \delta_{k,k'}
\end{align*}
and by the spectral theorem for self-adjoint operators these eigenvectors form a basis for any function $f$ in $L^2(\Omega,\pi)$, 
\be
f = \sum_{k = 0}^{|\Omega| -1} \langle f, f^k\rangle_\pi f^k. \label{eq:weighteddecomp}
\ee

Next we will use the basis \eqref{eq:weighteddecomp} to show that for any $f$ with $\mathbf{E}_\pi(f) = 0$ the spectral gap $\Delta_P$ satisfies
\begin{equation}
\Delta_P \leq \frac{\langle f , (I - P) f\rangle_\pi}{\langle f , f \rangle_\pi}.\label{eq:kequalsone}
\end{equation}
From \eqref{eq:weighteddecomp} we have 
\ba
\langle f , f \rangle_\pi &= \sum_{k , k' = 0}^{|\Omega|-1} \langle f , f^k \rangle_\pi \langle f^{k'} , f\rangle_\pi \langle f^k, f^{k'}\rangle_\pi \\
 &= \sum_{k = 0}^{|\Omega| - 1} |\langle f, f^k\rangle_\pi |^2, \label{eq:normsquare}
\ea
using the fact that $\{f^k\}$ are orthonormal in $L^2(\Omega,\pi)$.  Now since $\mathbf{E}(f)_\pi = 0$ we have $f = \sum_{k=1}^{|\Omega|-1} c_k f^k$,
with $\sum_{k=1}^{|\Omega|-1} |c_k|^2 = \langle f , f \rangle_\pi$, and so
\be
\langle f , (I - P) f\rangle_\pi = \sum_{k=1}^{|\Omega|-1} |c_k|^2 ( 1 - g_{\Omega,k}) \geq \langle f , f \rangle_\pi (1 - g_{\Omega,1}),
\ee
where in the last step we have used the fact that $G$ is positive semi-definite.

Next we will see that for any $f \in L^2(\Omega,\pi)$ the expression $\langle f, (I - P)f \rangle_\pi$ satisfies
\begin{equation}
\langle f, (I - P) f\rangle_\pi = \frac{1}{2} \sum_{x,y \in \Omega} \pi_x P_{x,y} |f_x - f_y|^2 .\label{eq:DirichletGen}
\end{equation}
Expanding the square in $\sum_{x,y \in \Omega} \pi_x P_{x,y} |f_x - f_y|^2$ yields
$$
\sum_{x,y \in \Omega} \pi_x P_{x,y} |f_x|^2-\sum_{x,y\in \Omega} \pi_x P_{x,y}(f_x f^*_y + f^*_x f_y) +\sum_{x,y\in \Omega}\pi_x P_{x,y} | f_y|^2 
$$
Since $\sum_{y\in \Omega} P_{x,y} = 1$ for  all $x\in \Omega$ and $\pi_x P_{x,y} = \pi_y P_{y,x}^*$ we see that the first and second terms are both equal to $\sum_{x\in \Omega} \pi_x |f_x|^2$, while the middle term can be simplified as,
$$
\sum_{x,y\in \Omega} \pi_x P_{x,y}(f_x f^*_y + f^*_x f_y)  = \sum_{x,y\in \Omega} \pi_y  f^*_y P_{y,x}^* f_x +\sum_{x,y\in \Omega}   \pi_x  f^*_x  P_{x,y} f_y = 2\langle f, P f\rangle_\pi
$$   
where the last step uses the fact that $P$ is self-adjoint with respect to the $\langle \cdot , \cdot \rangle_\pi$ product.  Therefore we have established the claim since
\begin{align*}
\sum_{x,y \in \Omega} \pi_x P_{x,y} |f_x - f_y|^2 = 2 \sum_{x\in \Omega} \pi_x |f_x|^2 - 2 \langle f , P f\rangle_\pi= 2 \langle f , (I - P) f\rangle_\pi.
\end{align*}

Finally, since $\mathbf{E}_\pi(f) = 0$ the norm $\langle f , f \rangle_\pi$ can be expressed in the (perhaps surprising) form of a variance,
\be
\langle f , f \rangle_\pi = \frac{1}{2}\sum_{x,y\in\Omega} \pi_x \pi_y |f_x - f_y|^2 , \label{eq:varianceDenominator}
\ee
because the cross terms vanish, $\sum_{x,y\in \Omega} \pi_x \pi_y f^*_x f_y = \sum_{x,y \in \Omega} \pi_x \pi_y f_x f^*_y = 0$.
Combining \eqref{eq:DirichletGen} with \eqref{eq:varianceDenominator} leads to the main result which is the culmination of this section,
\be
1 - p_1 = \min_{\substack{f \in \mathbb{C}^\Omega\\ f\perp \pi}} \frac{\sum_{x,y \in \Omega} \pi_x P_{x,y} |f_x - f_y|^2}{\sum_{x,y\in\Omega} \pi_x \pi_y |f_x - f_y|^2 }, \label{eq:dirichletFinal}
\ee
where the equality follows by taking $f$ to be the eigenstate $f^1$ with eigenvalue $p_1$ and using the equality \eqref{eq:DirichletGen}.

\subsection{Hamiltonian Conductance}\label{sec:HamiltonianConductance}
If $P$ is a reversible Markov chain with stationary distribution $\pi$ and state space $\Omega$, then for any subset $S \subset \Omega$ the quantity 
\begin{equation}
Q(S,\overline{S}) := \sum_{x \in S , y \in \overline S} \pi_x P_{x,y}\label{eq:regularconductance}
\end{equation}
can be thought of as the probability for the Markov chain to leave the set $S$.  The ratio $Q(S,\overline{S})/\pi(S)$ is called the conductance of the set $S$.  Generalizing the definition \eqref{eq:regularconductance} using the definitions of $\pi_x$ and $P_{x,y}$ from the previous sections,
\be
Q(S,\overline{S}) = \sum_{x \in S , y \in \overline S}\psi^*_x  G_{x,y} \psi_y = \langle \psi | \mathbbm{1}_S G 1_{\overline{S}} |\psi\rangle.  \label{eq:Q2}
\ee
Continuing to generalize the standard approach that is used when $P$ is a Markov chain, we now proceed to use Theorem 2 to see that the quantity in \eqref{eq:Q2} is non-negative.   Define the function
\begin{equation}
r^S_x = \begin{cases}
-\pi(\overline{S}) & \text{if } x\in S\\
\pi(S) & \text{if } x\in \overline{S},
\end{cases}
\end{equation}
which satisfies $\mathbf{E}_\pi(r^S) = 0$.  The norm $\|r^S\|_\pi$ satisfies
\begin{align}
\|r^S\|^2_\pi &= \sum_{x,y\in \Omega}\pi_x \pi_y \|r^S_x - r^S_y\|^2\\
&= \sum_{x\in S , y\in \overline{S}} \pi_x \pi_y \left( \pi(S)+\pi(\overline{S}) \right)^2 + \sum_{x \in \overline{S} , y \in S} \pi_x \pi_y \left(\pi(S) + \pi(\overline{S})\right)^2\\
&=2 \pi(S) \pi(\overline{S}) \label{eq:varianceeval}
\end{align}
Evaluating the numerator $\sum_{x,y \in S} \pi_x P_{x,y} \left(r^S_x - r^S_y\right)^2$ in \eqref{eq:dirichletFinal} yields,
\begin{align}
\sum_{x\in S,y \in \overline{S}} \pi_x P_{x,y} \left(\pi(\overline{S}) + \pi({S})\right)^2 +\sum_{x\in \overline S,y \in  S} \pi_x P_{x,y} \left(\pi(S) + \pi(\overline{S})\right)^2,
\end{align}
which can be further simplified as
\begin{align}
\sum_{x\in S , y\in \overline{S}} \pi_x P_{x,y} + \sum_{x\in \overline{S} ,y \in S} \pi_x P_{x,y} &= \sum_{x\in S , y\in \overline{S}} \psi^*_x G_{x,y} \psi_y + \sum_{x\in \overline{S} ,y \in S} \psi^*_x G_{x,y} \psi_y\\
&= \sum_{x\in S , y\in \overline{S}} \psi^*_x G_{x,y} \psi_y + \left( \sum_{x\in \overline{S} ,y \in S} \psi_x G^*_{x,y} \psi^*_y \right)^*\\
&= \sum_{x\in S , y\in \overline{S}} \psi^*_x G_{x,y} \psi_y +( \psi^*_x G_{x,y} \psi_y )^* \\
&= 2\Re\left[\langle \psi | 1_S G 1_{\overline{S} } | \psi\rangle \right]\label{eq:numnum}
\end{align}
where the $\Re$ in the final line denotes the real part of the enclosed expression.  

Note that $\bra{\psi}\mathbbm{1}_{S }G \mathbbm{1}_{S^c}\ket{\psi}$ is real and non-negative since
\begin{align}
\bra{\psi}\mathbbm{1}_{S} G \mathbbm{1}_{S^c}\ket{\psi} = \bra{\psi}\mathbbm{1}_{S} G \ket{\psi}) - \bra{\psi}\mathbbm{1}_{S} G |\mathbbm{1}_{S}\ket{\psi} =  \pi(S) - \langle \psi | 1_S G 1_S \rangle \geq \pi(S) - \|G\|\pi(S) = 0,\label{eq:conductancePositive}
\end{align}
and so combining \eqref{eq:varianceeval} with \eqref{eq:numnum} and \eqref{eq:dirichletFinal} yields
\be
\Delta_P \leq \min_{\substack{S \subseteq \Omega\\ 0 < \pi(S) < 1}} 2\frac{\langle \psi | 1_S G 1_{\overline{S}} |\psi\rangle}{\pi(S)\pi(\overline{S})}\label{eq:thirty}
\ee
Throughout this and the two preceeding sections we have been restricting the analysis to the subspace $\Omega \subseteq \mathcal{B}$ of basis vectors on which the ground state wave function is non-zero.  Let $A \subseteq \mathcal{B}$ be an arbitrary subset and define $S := A\cap \Omega$.  If $0 < \pi(A) < 1$,
\be
 \frac{\langle \psi | 1_A G 1_{\mathcal{B} - A} |\psi\rangle}{\pi(A)\pi(\mathcal{B} - A)}=\frac{\langle \psi | 1_{A\cap \Omega} G 1_{(\mathcal{B}-A)\cap \Omega} |\psi\rangle}{\pi(A \cap \Omega)\pi((\mathcal{B} - A)\cap \Omega)} .
\ee 
Now since $(\mathcal{B} - A)\cap \Omega  = \Omega - (A\cap \Omega)$,
\be
 \frac{\langle \psi | 1_{A\cap\Omega} G 1_{\mathcal{B} - A} |\psi\rangle}{\pi(A)\pi(\mathcal{B} - A)} =  \frac{\langle \psi | 1_A G 1_{\Omega - (A\cap \Omega)} |\psi\rangle}{\pi(A\cap\Omega)\pi(\Omega - (A\cap\Omega))} \geq \Delta_P
\ee
where the last step follows from \eqref{eq:thirty} because $S\subset \Omega$ with $0 < \pi(S) < 1$, and $\overline{S} = \Omega - A \cap \Omega$.  Now since $\Delta_P \geq \Delta_G$ we have shown
\begin{theorem}Under the same conditions of theorem \ref{theo:main} we have\label{theo:conductanceUB}
\begin{equation}
\Delta_H \leq \min_{\substack{S \subset \mathcal{B}\\ 0 < \pi(S) < 1}} \frac{-\langle \psi | \mathbbm{1}_S H \mathbbm{1}_{\bar{S}} |\psi\rangle}{\pi(S) \pi(\overline{S})}. \label{eq:conductanceUB}
\end{equation}
\end{theorem}
 To complete the proof of Theorem \ref{theo:main} we first note that
\begin{equation}
\min_{\substack{S \subset \mathcal{B}\\\pi(S)\neq 0}} \frac{\langle \psi | \mathbbm{1}_S G \mathbbm{1}_{S^c} |\psi\rangle}{\pi(S) \pi(S^c)} \leq 2 \min_{\substack{S \subset \mathcal{B}\\ 0< \pi(S) \leq 1/2}} \frac{\langle \psi | \mathbbm{1}_S G \mathbbm{1}_{S^c} |\psi\rangle}{\pi(S)},
\end{equation}
and from \eqref{eq:conductancePositive} we have
\begin{equation}
\langle \psi | \mathbbm{1}_S G \mathbbm{1}_{S^c} |\psi\rangle\leq \pi(\partial S).\label{eq:needtoshowMain}
\end{equation}

\subsection{Direct proof using the variational method}\label{sec:variational}
Define $\ket{\phi}:= \pi(S^c)\mathbbm{1}_{S}\ket{\psi}-\pi(S)\mathbbm{1}_{S^c}\ket{\psi}$ and note that $\braket{\phi}{\psi}=0$.  By the variational principle, \be2\bra{\phi}H\ket{\phi}-2E_0\braket{\phi}{\phi}\geq 2(E_1-E_0)\braket{\phi}{\phi}.\ee   Since $\braket{\psi}{\phi}=0$, it is clear that $\bra{\psi} H\ket{\phi}=0$ as well.  Using the elementary relations $\pi(S)+\pi(S^c)=1$ and $\mathbbm{1}_S+\mathbbm{1}_{S^c}=\mathbbm{1}$ we have that 
\begin{align*}
2\bra{\phi}H\ket{\phi}&=2\bra{\phi}H\ket{\phi}+(\pi(S)-\pi(S^c))\bra{\psi} H\ket{\phi}\\
&=2\bra{\phi}H\ket{\phi}+(\pi(S)-\pi(S^c))\bra{\psi}\mathbbm{1}_S H\ket{\phi}+(\pi(S)-\pi(S^c))\bra{\psi}\mathbbm{1}_{S^c} H\ket{\phi}\\
&=\bra{\phi}H\ket{\phi}+\bra{\phi}H\ket{\phi}+(\pi(S)-\pi(S^c))\bra{\psi}\mathbbm{1}_S H\ket{\phi}+(\pi(S)-\pi(S^c))\bra{\psi}\mathbbm{1}_{S^c} H\ket{\phi}\\
&=\bra{\phi}H\ket{\phi}+\pi(S^c)\bra{\psi}\mathbbm{1}_SH\ket{\phi}-\pi(S)\bra{\psi}\mathbbm{1}_{S^c}H\ket{\phi}\\ &\quad
+(\pi(S)-\pi(S^c))\bra{\psi}\mathbbm{1}_S H\ket{\phi}+(\pi(S)-\pi(S^c))\bra{\psi}\mathbbm{1}_{S^c} H\ket{\phi}\\
&=\bra{\phi}H\ket{\phi}+\pi(S)\bra{\psi}\mathbbm{1}_S H\ket{\phi}-\pi(S^c)\bra{\psi}\mathbbm{1}_{S^c} H\ket{\phi}\\
&=\pi(S^c)\bra{\psi}\mathbbm{1}_SH\ket{\phi}-\pi(S)\bra{\psi}\mathbbm{1}_{S^c}H\ket{\phi}+\pi(S)\bra{\psi}\mathbbm{1}_S H\ket{\phi}-\pi(S^c)\bra{\psi}\mathbbm{1}_{S^c} H\ket{\phi}\\
&=\bra{\psi}\mathbbm{1}_SH\ket{\phi}-\bra{\psi}\mathbbm{1}_{S^c}H\ket{\phi}\\
&=\bra{\psi}\mathbbm{1}_SH(\pi(S^c)\mathbbm{1}_{S}\ket{\psi}-\pi(S)\mathbbm{1}_{S^c}\ket{\psi})-\bra{\psi}\mathbbm{1}_{S^c}H(\pi(S^c)\mathbbm{1}_{S}\ket{\psi}-\pi(S)\mathbbm{1}_{S^c}\ket{\psi})\\
&=\bra{\psi}\mathbbm{1}_SH(\pi(S^c)\mathbbm{1}\ket{\psi}-\mathbbm{1}_{S^c}\ket{\psi})-\bra{\psi}\mathbbm{1}_{S^c}H(\mathbbm{1}_{S}\ket{\psi}-\pi(S)\mathbbm{1}\ket{\psi})\\
&=\pi(S^c)\pi(S)E_0+\pi(S^c)\pi(S)E_0-\bra{\psi}\mathbbm{1}_{S^c}H\mathbbm{1}_{S}\ket{\psi}-\bra{\psi}\mathbbm{1}_{S}H\mathbbm{1}_{S^c}\ket{\psi}\\
&=2E_0\pi(S)\pi(S^c)-2\bra{\psi}\mathbbm{1}_{S}H\mathbbm{1}_{S^c}\ket{\psi}
\end{align*}
where in the last line we use \eqref{eq:conductancePositive}. Next we have
\begin{align*}
2\bra{\phi}H\ket{\phi}-2E_0\braket{\phi}{\phi}&=2E_0\pi(S)\pi(S^c)-2\bra{\psi}\mathbbm{1}_{S}H\mathbbm{1}_{S^c}\ket{\psi}-2E_0\pi(S)\pi(S^c)\\
&=-2\bra{\psi}\mathbbm{1}_{S}H\mathbbm{1}_{S^c}\ket{\psi}
\end{align*}
so that applying the variational principle we obtain
$$-2\bra{\psi}\mathbbm{1}_{S}H\mathbbm{1}_{S^c}\ket{\psi}\geq2(E_1-E_0)\pi(S)\pi(S^c)$$
which can be rearranged to yield
$$\Delta_H \leq \frac{-\bra{\psi}\mathbbm{1}_{S}H\mathbbm{1}_{S^c}\ket{\psi}}{\pi(S)\pi(S^c)}.$$
This establishes the Hamiltonian conductance upper bound in theorem \ref{theo:conductanceUB}, and using the inequality \eqref{eq:needtoshowMain} completes the proof of \ref{theo:main}.  

\subsection{Generalization to higher eigenvalues}
In this section we complete the proof of theorem \ref{theo:main} by using the Courant-Fischer min-max characterization~\cite{horn2012matrix} of the higher eigenvalues of the self-adjoint operator $P$ with respect to the $\langle \cdot , \cdot \rangle_\pi$ inner product.  
\begin{lemma}The eigenvalues $1-p_k$ of $I-P$ can be expressed as a minimum over $(k+1)$-dimensional subspaces $U\subset \mathbb{C}^\Omega$,\label{lem:courant} 
\begin{equation}
1 - p_k = \min_{\substack{U \subset \mathbb{C}^\Omega\\ \textrm{dim}(U) = k+1}} \; \; \max_{\substack{g \in U\\g\neq 0}} \frac{\langle g , (I-P) g\rangle_\pi}{\langle g , g\rangle_\pi}.
\end{equation}
\end{lemma}

Applying lemma \ref{lem:courant} to the subspace $U$ spanned by a set of $k+1$ pairwise disjointly supported functions $g^0,...,g^k$ yields
\begin{equation}
1 - p_k \leq \max_{g\neq 0} \left \{ \frac{\langle g , (I - P) g\rangle_\pi}{\langle g , g \rangle_\pi} :g\in \textrm{span}\left(g^0,...,g^k\right) \right\}.
\end{equation}
If we further suppose that $g^0 ,..., g^k$ are supported on pairwise isolated sets, then for any $g = \sum_{i=0}^k \alpha_i g^i$ we have
\ba
\langle g , (I - P) g\rangle_\pi &= \frac{1}{2} \sum_{x,y \in \Omega} \pi_x P_{x,y}|g_x - g_y|^2\\
&= \frac{1}{2} \sum_{i=0}^k |\alpha_i|^2 \sum_{x,y \in \Omega} \pi_x P_{x,y}|g^i_x - g^i_y|^2.\label{eq:dirichletsubi}
\ea
Expressing the denominator as a variance and using the fact that all the terms in the sum are non-negative,
\be
\langle g, g \rangle_\pi =\frac{1}{2}\sum_{x,y\in\Omega} \pi_x \pi_y |g_x - g_y|^2 \geq \frac{1}{2}\sum_{i = 0}^k |\alpha_i|^2\sum_{x,y\in\Omega} \pi_x \pi_y |g^i_x - g^i_y|^2\label{eq:variancesubi}
\ee
Combining \eqref{eq:dirichletsubi} and \eqref{eq:variancesubi} yields
\ba
1 - p_k &\leq \max_{\alpha_0,...,\alpha_k} \frac{\sum_{i=0}^k |\alpha_i|^2 \sum_{x,y \in \Omega} \pi_x P_{x,y}|g^i_x - g^i_y|^2}{\sum_{i = 0}^k |\alpha_i|^2 \sum_{x,y\in\Omega} \pi_x \pi_y |g^i_x - g^i_y|^2}\\
 &\leq \max_{\substack{\alpha_0,...,\alpha_k\\\sum_{i=0}^k |\alpha_i|^2 = 1}} \; \; \sum_{i=0}^k |\alpha_i|^2\left( \frac{ \sum_{x,y \in \Omega} \pi_x P_{x,y}|g^i_x - g^i_y|^2}{ \sum_{x,y\in\Omega} \pi_x \pi_y |g^i_x - g^i_y|^2}\right)\\
&\leq \max_{i = 0,...,k}  \frac{\sum_{x,y \in \Omega} \pi_x P_{x,y}|g^i_x - g^i_y|^2}{\sum_{x,y\in\Omega} \pi_x \pi_y |g^i_x - g^i_y|^2}\\
&= \max_{i = 0,..., k} \frac{\langle g^i , (I-P) g^i\rangle_\pi}{\langle g^i , g^i\rangle_\pi}.
\ea
Now for the special case in which $g^i = \mathbbm{1}_{S_i}$ is the indicator function for a subset $S_i$ we can apply the same steps used in sections \ref{sec:HamiltonianConductance} to obtain
\ba
1 - p_k &\leq \max_i \frac{\langle \psi | \mathbbm{1}_{S_i} G \mathbbm{1}_{\overline{S}_i}| \psi \rangle}{\pi(S_i)\pi(\overline{S}_i)} \label{eq:excitedconductance},
\ea
and applying the definition of $G$ in terms of $H$,
\be
E_k  - E_0 \leq \max_i \frac{-\langle \psi | \mathbbm{1}_{S_i} H \mathbbm{1}_{\overline{S}_i}| \psi \rangle}{\pi(S_i)\pi(\overline{S}_i)}
\ee
and so assuming $\pi(S_i) \leq 1/2$ for all $i$ this becomes
\ba
E_k  - E_0  &\leq 2(\|H\| - E)\max_{i}  \frac{\pi(\partial S_i)}{\pi(S_i)},\label{eq:excitedConductance}
\ea
which completes the proof of \eqref{eq:multiexpansion}.  Note that this proof is substantially based on the one given in~\cite{lee2014multiway} for unweighted graph Laplacians, with the main difference being that the stronger stipulation that the subsets $S_i$ be isolated was necessary to obtain \eqref{eq:dirichletsubi}.  

\section{Examples and applications}\label{sec:examples}

\paragraph{Ferromagnetic Transverse Ising Model.}
By tuning the transverse field of the ferromagnetic quantum Ising model we can demonstrate two extremes of the inequality \eqref{eq:expansion}.  The Hamiltonian is
\be
H = -\Gamma \sum_{i=1}^n   X_i-  \alpha \sum_{i=1}^n Z_i Z_{i+1},\label{eq:tim}
\ee
First we establish that the bound \eqref{eq:expansion} is asymptotically tight when $\alpha= 0$ and $\Gamma \neq 0$, which shows that the factor of $\|H\|$ in the bound is necessary.  This follows because the ground state of $H$ in the computational basis is the uniform superposition, so the set $S$ of basis states with Hamming distance of $k$ or less contains $|S| = 1 + \binom{n}{1} + ... + \binom{n}{k}$ points, while the interior boundary has size $|\partial S| = \binom{n}{k}$.  Therefore $\pi(\partial S)/\pi(S) = |\partial S| / |S|$ is $\mathcal{O}(k^{-1})$, and since $\|H\| = n$ and $\Delta_H = 1$ in this case taking $k = \lfloor n/2 -1 \rfloor$ shows that \eqref{eq:expansion} is asymptotically tight.\footnote{This example is well known in the context of Markov chains since $I - H/\|H\|$ is the transition matrix for a uniform random walk on the hypercube.  We present it here to emphasize that the factor of $\|H\|$ is necessary in general when dealing with Hamiltonians.}

Having seen that the paramagnetic phase where $\Delta_H$ is constant is associated with a ground state that is highly expanding in the computational basis, we next turn to the ferromagnetic phase when the gap between the ground state and first excited state is exponentially small.  Again we let $\mathcal{B} = \{0,1\}^n$ be the computational basis, and we define $S$ to be the set of computational basis states with a negative expectation value of the magnetization operator $M := \sum_{i=1}^n Z_i$,
\be
S = \{z \in \{0,1\}^n : \langle z | M | z\rangle < 0\}.
\ee
Since $z,z'\in \mathcal{B}$ with $z \neq z'$ have $\langle z | H | z'\rangle \neq 0$ iff $|\langle z | M | z \rangle - \langle z' | M | z'\rangle | = 1$, the interior boundary of $S$ is 
\be
\partial S = \{ z \in \{0,1\}^n : \langle z | M | z\rangle = -1 \}.
\ee
From the symmetry of $H$ we know that $\pi(S)$ is $\Omega(1)$, and in the ferromagnetic phase the probability of measuring the magnetization of the system to be near zero is $\mathcal{O}(e^{-n})$.  This implies that $\pi(\partial S)$ is $\mathcal{O}(e^{-n})$, which demonstrates that \eqref{eq:expansion} can be used to show that the spectral gap of this system is exponentially small.

\paragraph{Ground state close to a GHZ state.}
The $n$ qubit GHZ state exhibits long range entanglement,
$$
|\psi_{GHZ}\rangle = \frac{1}{\sqrt{2}} \left( |00...0\rangle + |11...1\rangle \right).
$$
Let $H$ be a $k$-local Hamiltonian with a ground state $|\psi\rangle$ that is close to $|\psi_{GHZ}\rangle$, 
\be
\operatorname{tr}\left[ |\psi\rangle\langle \psi| - |\psi_{GHZ}\rangle \langle \psi_{GHZ}| \right] < \epsilon\label{eq:close}
\ee
Choose the basis $\mathscr{B} = \{0,1\}^n$ and set $\pi(S) = \sum_{x\in S} \|\psi_x\|^2$ for $S\subseteq{B}$, then \eqref{eq:close} implies
$$
\pi(\mathscr{B} - \{00...0 , 11...1\}) < \epsilon.
$$
For $x\in \mathscr{B}$ let $|x|$ denote the number of 1s in the string.  Define 
\be
B_r := \{x \in \{0,1\}^n : |x| \leq r\}.
\ee
 Without loss of generality we may have $1/2 - \epsilon < \pi(B_{j k}) \leq 1/2$, where if the upper bound were not satisfied we would center the Hamming balls on 11...1  instead of 00...0.  
 
Observe that the interior boundary of $B_{2k}$ is disjoint from $B_k$ , $\partial B_{2 k} \cap B_k = \O$.  This pattern generalizes, so for $j,j' = 1 , ... , \lfloor n/k\rfloor$ we have $j \neq j' \Rightarrow \partial B_{j k} \cap B_{j' k} = \O$. This disjointness implies
 \begin{align}\label{eq:ghzbound}
\Delta_H &\leq 2 (\|H\| - E)\min \left \{\frac{\pi\left(\partial B_{j k}\right)}{\pi(B_{j k})} : j = 1,..., \lfloor n/k \rfloor \right \}\\
& \leq \frac{4 \epsilon (\|H\| - E)}{\lfloor n/k\rfloor (1 - \epsilon)},
\end{align}
and so any local Hamiltonian which has a ground state within $\epsilon$ of the GHZ state has an $\mathcal{O}(\epsilon)$ spectral gap.  

\paragraph{Inherently gapless ground states.}
In condensed matter physics one distinguishes between ``gapped'' and ``gapless'' systems, and this distinction does not necessarily depend on the spectral gap $\Delta_H := E_1 - E_0$ between the ground state and the first excited state.  A family of Hamiltonians is called gapped if in the limit of an infinite system size the ground state has a finite degeneracy, and if there is a finite gap $\Delta$ to the energy excitations above this degenerate ground state~\cite{zeng2015quantum}.  If the system is not gapped then it is called gapless.  

For example, the ferromagnetic transverse Ising model \eqref{eq:tim} is known to be gapped when $\alpha =1$ and $\Gamma \neq 1$, despite the fact that the spectral gap $\Delta_H$ is exponentially small in the system size when $\Gamma < 1$.  This exponentially small splitting goes to zero in the thermodynamic limit, and though the ground state is then doubly degenerate it turns out that there is always a constant gap to the third energy eigenstate when $\Gamma \neq 1$.  

Ground states of gapped systems have exponentially decaying correlation functions~\cite{hastings2006spectral},  and are strongly believed to have only area law scaling of entanglement (although this is only known rigorously for 1D systems~\cite{hastings2007area}).  From properties of this sort there is a tendency to speak of ground states themselves as gapped or gapless, although strictly speaking gapped vs gapless is a property belonging to the parent Hamiltonians which have these as ground states.  Here we note that a consequence of the bound \eqref{eq:multiexpansion} is that in principle it enables one to show that a ground state is intrinsically gapless with respect to any local Hamiltonian by identifying a number of localized probability clusters in some basis such that the number of clusters scales with the system size.  

\paragraph{History-state Hamiltonians are inherently gapless.}\footnote{See \cite{cubittCritical} for a more general proof of this fact as well as additional results related to the low-energy spectrum of history-state Hamiltonians.  That work inspired the application of our results to this example.}
In the field of Hamiltonian complexity, Kitaev~\cite{kitaev2002classical} repurposed a Hamiltonian that was originally proposed by Feynman~\cite{feynman1986quantum} in order to show that quantum circuits can be mapped onto ground states of local Hamiltonians.  Given a quantum circuit specified by local unitary gates $U_1,...,U_T$, the idea is to entangle the successive time steps of the circuit with an ancillary register called the time register in order to produce states of the form,
\be
|\psi_{\textrm{hist}}\rangle = \frac{1}{\sqrt{T+1}}\sum_{t = 0}^T U_t ... U_0 |0^n\rangle|t\rangle.\label{eq:historystate}
\ee
These states are called history states because they contain each step in the history of the quantum circuit in a superposition.  A frustration-free Hamiltonian with this state as its unique ground state is ${H:= \left( \sum_{i = 1}^n |1\rangle \langle 1|_i \right) \otimes |0\rangle \langle 0| + \sum_{t=0}^T H_{\textrm{circuit}}(t)}$, with
\ba
H_{\textrm{circuit}}(t) &:= \frac{1}{2}\left (I\otimes |t\rangle \langle t |  + I\otimes |t -1 \rangle \langle t - 1|  - U_t \otimes |t \rangle \langle t- 1| t  - U^\dagger_t\otimes |t - 1\rangle \langle t | \right).\label{eq:hcircuit}
\ea
If the size of the circuit scales polynomially in $n$, then the local terms of the form \eqref{eq:hcircuit} are $\log(n)$-local.  By using time registers encoded in unary along with perturbative gadgets to reduce the locality of the interactions, a large number of works in Hamiltonian complexity have been devoted to proving that history states of the form \eqref{eq:historystate} can appear as ground states of Hamiltonians with few-body interactions restricted to be geometrically local on low-dimensional lattices.  This motivates one to consider the spectral properties inherent to all local Hamiltonians which have ground states that are effectively of the form \eqref{eq:historystate}.  Suppose we consider generalized history state Hamiltonians $H$ with ground states of the form \eqref{eq:historystate} that do not only connect time steps which are nearest-neighbors, but also contain terms connecting time steps that are further apart?  Recent examples of such generalizations include uniform distributions on more general graphs~\cite{bausch2016complexity} as well as history states with modified amplitudes~\cite{Bausch2018analysislimitations}.   We consider the class of history state Hamiltonians with ground states of the form \eqref{eq:historystate} which are ``temporally $k$-local'', i.e. if $|\psi_t\rangle := U_t ... U_0 |0^n\rangle$ then
$$
\langle t| \langle \psi_t|  H | \psi_{t'}\rangle |t'\rangle = 0 \textrm{ if } |t-t'| > k.  
$$  
for some $k$ that is $\mathcal{O}(1)$.  We assume $H = \sum_{t = 0}^T H_t$ with $\max_{t} \|H_t\| = \mathcal{O}(1)$, but the $H_t$ are not assumed to be projectors as in the standard construction. The restriction that the Hamiltonian contains $T+1$ local terms is without loss of generality for temporally $k$-local Hamiltonians, since the form of \eqref{eq:historystate} means there are no additional qubits besides those representing the computation and the clock, and $k$ can be chosen large enough to include all of the local terms that act on time step $t$. 

To apply our methods we consider the set of basis labels $\mathcal{B} = \{0,...,T\}\times \{0,1\}^n$.  Define the subset $S = \{1,...,T/4\}\times \{0,1\}^n$, which for the ground state \eqref{eq:historystate} satisfies $\pi(S) = 1/4$.   The numerator of the Hamiltonian conductance in theorem \ref{theo:conductanceUB} satisfies
$$
-\langle \psi_{\textrm{hist}} | \mathbbm{1}_S H \mathbbm{1}_{\bar{S}} | \psi_{\textrm{hist}}\rangle \leq \left(\max_{t} \|H_t\| \right)\sum_{t = \frac{T}{4} - k}^{\frac{T}{4}} \sum_{t' = \frac{T}{4} +1}^{\frac{T}{4} + k}|\langle \psi_{\textrm{hist}}|t\rangle|  |\langle t' | \psi_{\textrm{hist}}\rangle| = \mathcal{O}(T^{-1}),
$$
and therefore by theorem \ref{theo:conductanceUB} we have $\Delta_H = \mathcal{O}(T^{-1})$.  This shows that for any temporally $k$-local history state Hamiltonian with a ground state of the form \eqref{eq:historystate}, and suitably normalized local terms, necessarily has an excited state within energy $\mathcal{O}(T^{-1})$ of the ground state.  This argument can also be extended to produce additional excited states by taking 
$$S_i = \{i T^{1/4} + 1, ... , (i+1) T^{1/4} -1\} \times \{0,1\}^n \textrm{ for } i = 1,...,T^{3/4}.$$  This yields $\pi(S_i) = \Omega(T^{1/4})$ for each $i$, all of which satisfy $-\langle \psi_{\textrm{hist}} | \mathbbm{1}_{S_i} H \mathbbm{1}_{\bar{S_i}} | \psi_{\textrm{hist}}\rangle = \mathcal{O}(T^{-1})$, and so by \eqref{eq:excitedconductance} this shows that $H$ has $\Omega(T^{3/4})$ excited states with energies that are $\mathcal{O}(T^{-3/4})$, and therefore temporally local Hamiltonians with ground states of the form \eqref{eq:historystate} are inherently gapless.

\paragraph{Exponentially small gap in the generalized Motzkin chain.}
Motzkin spin chains are a class of analytically tractable class of 1D Hamiltonians which are intringuing because they posess long range entanglement, even though they are frustration-free local Hamiltonians.   Originally proposed as a spin $s = 1$ model with an inverse polynomial scaling of the spectral gap and entanglement entropy for the half chain region scaling like $S = \mathcal{O}(\log(n))$~\cite{bravyi2012criticality}, the so-called colored Motzkin models with $s > 1$ were later found to yield power law entanglement $S = \mathcal{O}(\sqrt{n} \log s)$~\cite{movassagh2016supercritical}, and recently these models have been further generalized~\cite{zhang2016quantum} to produce volume entanglement scaling $S = \mathcal{O}(n \log(s))$ for the half chain region.

Following the introduction of the generalized Motzkin ground states with extensive entanglement entropy it was shown that these chains have an exponentially small spectral gap~\cite{levine2016gap}.  The proof uses the fact that these spin chains are stoquastic and frustration-free, and so by the quantum-to-classical mapping they correspond to Markov chains and the standard vertex expansion upper bound for the spectral gap of a Markov chain is applied to show that the Hamiltonian spectral gap is exponentially small.  Our reason for reviewing this work here is to point out that the mapping to a Markov chain is not necessary when one only seeks an upper bound on the gap using the vertex expansion of the ground state probability distribution.  In fact, by using the properties of the ground state probability distribution that are established in section D of \cite{levine2016gap} together with our Theorem \ref{theo:main} one sees that any local Hamiltonian (not necessarily stoquastic, or frustration-free) whose ground state probability distribution matches that of the generalized Motzkin chains necessarily has an exponentially small gap.
\section{Irreducible stoquastic Hamiltonians}\label{sec:stoquastic}
In this section we will assume that $H$ is a stoquastic Hamiltonian in the basis $\mathcal{B}$, and that the matrix form of $H$ in this basis is irreducible.  A symmetric matrix is irreducible if and only if replacing the non-zero entries of the matrix with 1 yields the adjacency matrix of a connected graph.   Define $\Gamma$ to be the magnitude of the least non-zero off-diagonal matrix element of $H$,
\be
\Gamma = \min_{z,z' \in \mathcal{B} : \; z \neq z' \; \textrm{and} \; H_{z,z'} \neq 0} |\langle z | H | z'\rangle|.\label{eq:Gamma}
\ee
The main result of this section is a vertex expansion lower bound for the spectral gap that is useful when $\Gamma$ is at most inverse polynomially small.  An important class of systems that these bounds apply to are adiabatic optimization Hamiltonians with a transverse field driver term,
\be
H(s) = (1-s)\sum_{i =1}^n \frac{1}{2}\left(I - X_i\right) \; + \; s \sum_{z \in \{0,1\}^n} f(z) |z\rangle \langle z|.
\ee
If we assume $\Delta_{\min} = \min_{0\leq s < 1} \Delta_{H(s)}$ is at most inverse polynomially small, then the annealing protocol can conclude at $s^* = 1 - 1/\textrm{poly}(n)$ since the ground state $\psi(s^*)$ will have a high overlap with the final state $\psi(1)$.  One may also apply these results to adiabatic optimization with other irreducible stoquastic driver terms besides the transverse field (including any driver Hamiltonians which have the uniform transverse field as a subset of the Hamiltonian terms).  

We begin by defining $G$ and $P$ as before,
\begin{equation}
G = \frac{\|H\| - H}{\|H\| - E} \quad , \quad P = \sum_{x,y\in \Omega} \frac{\psi_y}{\psi_x} G_{xy} |x\rangle \langle y|,
\end{equation}
and now since $H$ is stoquastic we have that $G$ is a matrix with non-negative entries in the basis $\mathcal{B}$.  Since $|\psi\rangle$ is the principal eigenvector of $G$ it is also guaranteed to have non-negative components by the Perron-Frobenius theorem, and so $P_{xy} \geq 0$ for all $x,y\in \Omega$.  Since we have already verified that $\sum_{y \in \Omega}P_{xy} = 1$ for all $x \in \Omega$, we that $P$ is a stochastic matrix in the computational basis.\footnote{
Although $P$ is always a Markov chain transition matrix when $H$ is stoquastic, the transition probabilities depend ratios of the form $\psi_y / \psi_x$, and so the Markov chain cannot be efficiently simulated on a classical computer unless these ratios are known to high precision.  For the special case of frustration-free stoquastic Hamiltonians there is an efficient method for computing these ratios, which was used to obtain a polynomial time classical simulation algorithm for frustration-free stoquastic adiabatic computation~\cite{bravyi2009complexity}.  The fact that $P$ still defines a Markov chain when $H$ is frustrated was noted in~\cite{bravyi2009complexity}, but limitations on generalizing this idea are also known~\cite{somma2013spectral}. However, see \cite{aharonov2008adiabatic, bausch2016increasing} for examples in which this mapping was applied to frustrated Hamiltonians in order to analyze the spectral gap using Markov chain methods.}  Furthermore, a variant of the Perron-Frobenius theorem for irreducible non-negative matrices implies that $|\psi\rangle$ has all positive amplitudes in the chosen basis.  To see this suppose that $\Omega \subset \mathcal{B}$, then since $G$ is irreducible there must be some pair $(x,y)$ with $x\in \Omega - \mathcal{B}$ and $y \in \Omega$ with $G_{x,y} > 0$.  But this is a contradiction because
\be
\langle x | G | \psi\rangle = \sum_{y'\in \Omega} G_{x,y'}\psi_{y'},
\ee 
and since $\psi_{y'} > 0$ for all $y' \in \Omega$ this implies $\langle x | G | \psi \rangle = E \langle x | \psi\rangle > 0$.
From \eqref{eq:Gamma} we have that the non-zero off-diagonal elements of $G$ are no smaller than $\Gamma/\|H\|$.  This is not enough to lower bound the transition probability $P_{x,y}$, however, because so far we do not have any lower bound on the ratio $\psi_y / \psi_x$.  To bound this ratio note that
\be
\sum_{x'\in \Omega} P_{y,x'} = 1 \; \textrm{and} \; P_{y,x'} > 0 \implies P_{y,x} \leq 1,
\ee 
and since $P_{y,x} = \psi_x G_{x,y} \psi^{-1}_y \geq \psi_x \psi^{-1}_y \Gamma/\|H\|$, we must have $\psi_x/\psi_y \leq \|H\|/\Gamma$, which can be rearranged to yield
\be
\frac{\psi_y}{\psi_x} \geq \frac{\Gamma}{\|H\|} \implies P_{x,y} \geq \left(\frac{\Gamma}{\|H\|}\right)^2, \label{eq:lipschitz}
\ee
which holds for all $x,y$ with $H_{x,y} \neq 0$.  

Applying Cheeger's inequality for Markov chains to $P$, we have
\be
\frac{\Phi^2}{2} \leq \Delta_P \quad , \quad \Phi := \min_{\substack{S \subseteq \Omega\\0 < \pi(S) \leq 1/2}} \; \;\frac{1}{\pi(S)} \sum_{x\in S , y \in \overline{S}} \pi_x P_{x,y}. \label{eq:markovCheeger}
\ee
Using \eqref{eq:lipschitz} we have
\be
 \sum_{x\in S , y \in \overline{S}} \pi_x P_{x,y} = \sum_{x \in \partial S ,y \in \overline{S}} \pi_x P_{x,y} \geq \left(\frac{\Gamma}{\|H\|}\right)^2 \sum_{x \in \partial S} \pi_x
\ee
and so $\Phi(S) \geq \left(\frac{\Gamma}{\|H\|}\right)^2 \pi(\partial S)/\pi(S)$ for all $S\subseteq \Omega$.  Now since 
$$
\Delta_P = \Delta_{\mathbbm{1}_\Omega G \mathbbm{1}_\Omega}= \Delta_G = \frac{\Delta_H}{\|H\| - E} , 
$$ 
we can state a version of Cheeger's inequality for irreducible stoquastic Hamiltonians.   
\begin{theorem}\label{thm:stoquastic}
If $H$ is stoquastic and irreducible in the basis $\mathcal{B}$, and all the non-zero off-diagonal matrix elements of $H$ have magnitude at least $\Gamma$, then
\be
\min{\substack{S \subset \mathcal{B} \\ 0 < \pi(S) \leq 1/2}} \; \frac{\Gamma^4}{2\|H\|^4\left(\|H\| - E\right)} \left(\frac{\pi(\partial S)}{\pi(S)}\right)^2 \leq \Delta_H
\ee
\end{theorem}

\section{\large Counterexample to a non-stoquastic lower bound}\label{sec:counterexample}
In this section we give a counterexample to the most straightforward generalization of Cheeger's inequality \eqref{eq:markovCheeger} that one might propose for non-stoquastic Hamiltonians\footnote{Note that generalizations of Cheeger's inequality for signed graphs~\cite{gournay2016isoperimetric, martin2017frustration} and general Hermitian matrices~\cite{jarret2018hamiltonian} have been published since the first appearance of this work.  As expected, these inequalities differ from the naive generalization which we give a counterexample against here.}.  We will show that the inequality 
\be
\frac{\Phi^2}{2} \leq \Delta_P \quad , \quad \Phi := \min_{\substack{S \subseteq \Omega\\0 < \pi(S) \leq 1/2}} \; \;\frac{\langle \psi | \mathbbm{1}_S G \mathbbm{1}_{\overline{S}} | \psi\rangle}{\langle \psi | \mathbbm{1}_S | \psi\rangle}, \label{eq:markovCheeger}
\ee
which always holds when $G$ is defined from a stoquastic Hamiltonian, can be violated by an exponentially large factor when $G$ is defined from a non-stoquastic Hamiltonian.

Consider a single particle on a 1-dimensional ring with $2n+1$ sites, labeled $-n,-n+1,\ldots,n$.  The Hamiltonian takes the form of a discrete momentum operator 
\begin{equation}
\bra{i}H_0\ket{j} = \delta_{i,j} + \frac{1}{2}\left( \delta_{i+1,j}+\delta_{i-1,j}\right) + \frac{1}{2}\left(\delta_{n,-n} + \delta_{-n,n} \right)
\end{equation}
for $i,j=-n,\ldots n$.  Its (normalized) eigenstates are momentum eigenstates $\ket{p_j}$ defined by $$\braket{j}{p_k}=\frac{1}{2n+1}\exp\left(2\pi i \frac{jk}{2n+1}\right).$$  Its eigenvalues are $$E_k=1+\cos\left(\frac{2\pi k}{2n+1}\right).$$  

Clearly the two lowest eigenvalues are for values of $k$ that make $2\pi k/(2n+1)$ as close to $\pi$ as possible.  There are two that are equally close to $\pi$, $k=-n,n$.  By symmetry these two energy levels are degenerate and hence $H_0$ has a zero gap.  As a counterexample to the naive lower bound we might hope to hold for nonstoquastic Hamiltonians as well, we will first show that this is in a sense already a counterexample, because the minimum conductance is $\mathcal O(poly(n))$.  Then we will perturb the Hamiltonian to have an exponentially small gap and demonstrate a similar polynomially small conductance.  

Let $\ket{\psi}=\ket{p_n}$ be the ground state of $H_0$.  The conductance of $H_0$ is just 
\begin{align}
-\bra{\psi}\mathbbm{1}_S H_0 \mathbbm{1}_{S^c}\ket{\psi}&=\frac{1}{2n+1}\left(\sum_{j\in S\cap (\bar S+1)}e^{2\pi in \frac{j-1}{2n+1}}e^{-2\pi in \frac{j}{2n+1}}+\sum_{j\in S\cap (\bar S-1)}e^{2\pi in \frac{j+1}{2n+1}}e^{-2\pi in \frac{j}{2n+1}}\right)\\
&=\frac{1}{2n+1}\left(\sum_{j\in S\cap (\bar S+1)}e^{2\pi in \frac{-1}{2n+1}}+\sum_{j\in S\cap (\bar S-1)}e^{2\pi in \frac{1}{2n+1}}\right)\\
&=\frac{1}{2n+1}\left(\abs{S\cap (\bar S+1)}e^{2\pi in \frac{-1}{2n+1}}+\abs{S\cap (\bar S-1)}e^{2\pi in \frac{1}{2n+1}}\right)\\
&=\frac{\cos({2\pi n \frac{1}{2n+1}})}{2n+1}\left(\abs{S\cap (\bar S+1)}+\abs{S\cap (\bar S-1)}\right)
\end{align}
where the last line follows from the conductance being a real quantity (maybe include that brief proof of this fact from before).  It must be minimized when the sizes of the two sets are as small as possible.  Neither can be zero, so the smallest they can both be is 1.  Hence our minimal conductance goes like $\mathcal O(n^{-2})$.  Our gap however is zero.  

In order to give a counterexample with a unique ground state, we add a perturbation $\delta H$ to $H_0$ 
\begin{equation}
\langle i|\delta H| j\rangle = 2^{-n} \left(\delta_{-n,i} \delta_{n-1 ,j}+\delta_{-n,j} \delta_{n-1,i}\right).
\end{equation}
We will compute the matrix elements of $\delta H$ in the ground subspace of $H$, spanned by $\ket{p_n}$ and $\ket{p_{-n}}$.  We can compute for large $n$ that
\ba\bra{p_{-n}}\delta H\ket{p_{n}}&=2^{-n}\left(e^{2\pi i n(n-1)/(2n+1)}e^{2\pi i n^2/(2n+1)}+e^{-2\pi i n^2/(2n+1)}e^{-2\pi i n(n-1)/(2n+1)}\right)\\
&=2^{-n}\cos\left(2\pi\frac{2n^2-1}{2n+1}\right)\\
&
\approx 1\ea
The diagonal elements are
\ba\bra{p_{\pm n}}\delta H\ket{p_{\pm n}}&=2^{-n}\left(e^{\pm 2\pi i (n-1)n/(2n+1)}e^{\pm2\pi i n^2/(2n+1)}+e^{\mp 2\pi i (n-1)n/(2n+1)}e^{\mp 2\pi i n^2/(2n+1)}\right)\\
&=\pm 2^{-n}\cos\left(2\pi\frac{2n^2-1}{2n+1}\right)\ea
Meaning that to zeroth order the gap of $H_0+\delta H$ is $\mathcal O(2^{-n})$.  The perturbation, however, only contributes an exponentially small factor to the conductance.  Hence here the conductance lower bound is still $\mathcal O(n^{-2})$ but the gap is $\mathcal O(2^{-n})$, and so we have exhibited a non-stoquastic Hamiltonian with a unique ground state that violates the bound \eqref{eq:markovCheeger} by an exponentially large multiplicative factor.  
\section{Discussion of adiabatic computing}\label{sec:adiabatic}
A primary motivation for this work is to understand the benefits as well as potential limitations of using non-stoquastic driver Hamiltonians in adiabatic optimization.  General adiabatic computation with non-stoquastic Hamiltonians can implement arbitrary quantum circuits with only polynomial overhead, however in this section we contrast universal adiabatic computation with more restricted notions of adiabatic optimization and ground state sampling that can also be considered.  
\paragraph{Universal adiabatic computation.} First we briefly review universal AQC~\cite{aharonov2007adiabatic}.  The goal is to implement a quantum circuit consisting of local gates $U_0, U_1 , ... , U_T$ acting on an input state that may be taken without loss of generality to be the all zeroes computational basis state $|0^n\rangle$.  Define 
\be
H_{\textrm{init}}:= \left( \sum_{i = 1}^n |1\rangle \langle 1|_i \right)\otimes |0\rangle \langle 0 | \quad , \quad H_{\textrm{final}} := H_{\textrm{init}} +  \sum_{t=0}^T H_{\textrm{circuit}}(t) \quad , \quad
H(s) := s H_\textrm{init} + (1-s) H_{\textrm{final}}. \nonumber
\ee
where $H_{\textrm{circuit}}(t)$ is given in \eqref{eq:hcircuit}.  The ground state of $H(s)$ at $s = 0$ is simple to prepare, and the ground state at $s = 1$ will be the history state of the circuit which is given in \eqref{eq:historystate}.  

In \cite{aharonov2007adiabatic} it was proven that the minimum spectral gap $\Delta_H := \min_{0\leq s \leq 1} \Delta_{H(s)}$ throughout the evolution is $\Omega(T^{-3})$.   Notice that the probability of measuring $|\psi_{\textrm{history}}\rangle$ in the final step $t = T$ of the computation is $\Omega(T^{-1})$.  A standard trick is to pad the end of the circuit with identity gates to increase the probability of measuring the final step of the computation.  If we add $r$ identity gates at the end of the circuit, the probability of measuring $| \psi_{\textrm{history}}\rangle$ to be in one of the last $r$ steps of the computation increases to $\Omega(r/(r+T))$, while the spectral gap only falls to $\Omega\left(1/(r + T)^3\right)$.  

Now we can see how the bound \eqref{eq:expansion} is consistent with ability of universality of adiabatic computation to produce output states with low expansion, e.g. the GHZ state from section \ref{sec:examples}.  The statment that AQC is universal only requires the ground state to have non-neglible $1/\textrm{poly}(n)$ overlap with the output state.  The GHZ state can be produced using $T = \mathcal{O}(\log(n))$ gates, however, if we instead require the ground state of $H(1)$ to have overlap $1 - \epsilon$ on the GHZ state then the circuit must be padded with $r = T / (1 - \epsilon)$ extra gates.  This padding reduces the spectral gap $\Omega\left(1/(r + T)^3\right)$ so that the bound \eqref{eq:ghzbound} remains satisfied.  

To summarize, the Solovay-Kitaev theorem states that arbitrary an arbitrary unitary $U$ can be implemented to precision $\epsilon$ using a number of gates that grows as $\textrm{poly}(\log \epsilon^{-1})$, but using pure adiabatic evolution without measurement there are states which provably take time $\Omega(\epsilon^{-1})$ to produce to accuracy $\epsilon$.   In universal AQC this is resolved by measuring the clock register once the overlap is guaranteed to be sufficiently high, and if the measurement succeeds then the system will be projected into the exact output of the circuit.  

\paragraph{Competing local minima in adiabatic optimization.}
A key difference between universal AQC and the standard form of adiabatic optimization (QAO) is that as $s \to 1$ the latter approaches a Hamiltonian that is diagonal in the computation basis,
\be
H_P  = \sum_{z\in \{0,1\}^n} f(z)|z\rangle\langle z|.
\ee
A commonly cited failure mode of transverse field adiabatic optimization occurs when the cost function $f$ has near degenerate local minima corresponding to computational basis states that are far separated in Hamming distance~\cite{boixo2014evidence}.  In some cases these observations have been rigorously shown to cause transverse field adiabatic optimization to generically encounter exponentially small spectral gaps for certain classes of NP-hard cost functions~\cite{altshuler2009adiabatic, altshuler2010anderson,laumann2015quantum}.  The bound \eqref{eq:expansion} suggests that even with non-stoquastic driver Hamiltonians this mode of failure cannot be fully avoided as long as the ground state probability distribution becomes concentrated on far separated local minima of the cost function as $s \to 1$.  

\paragraph{Example: geometry of a non-stoquastic path change.}
Following the preceding sections we've seen that the success of including ``path change'' terms in the Hamiltonian depends on removing bottlenecks from the ground state probability distribution throughout the adiabatic path.  Here we illustrate this feature for a particular case of a successful non-stoquastic path change in a bit-symmetric toy model~\cite{farhi2002quantumB}.  

This path change was analyzed in~\cite{farhi2002quantumB} using the spin coherent potential, which has degenerate local minima at $s = 0.434...$ without the path change, but has a single global minimum throught the interval $0 \leq s \leq 1$ when the path change term is present.  The low lying spectrum and the ground state probability distribution of the system as a function of Hamming weight (which were not exhibited in \cite{farhi2002quantumB}) with and without the path change  are shown in Figures \ref{fig:spectraPathChange} and \ref{fig:gsPathChange}.  
\begin{figure}[h]
\begin{center}
\includegraphics[scale=1,trim=20 0 0 0]{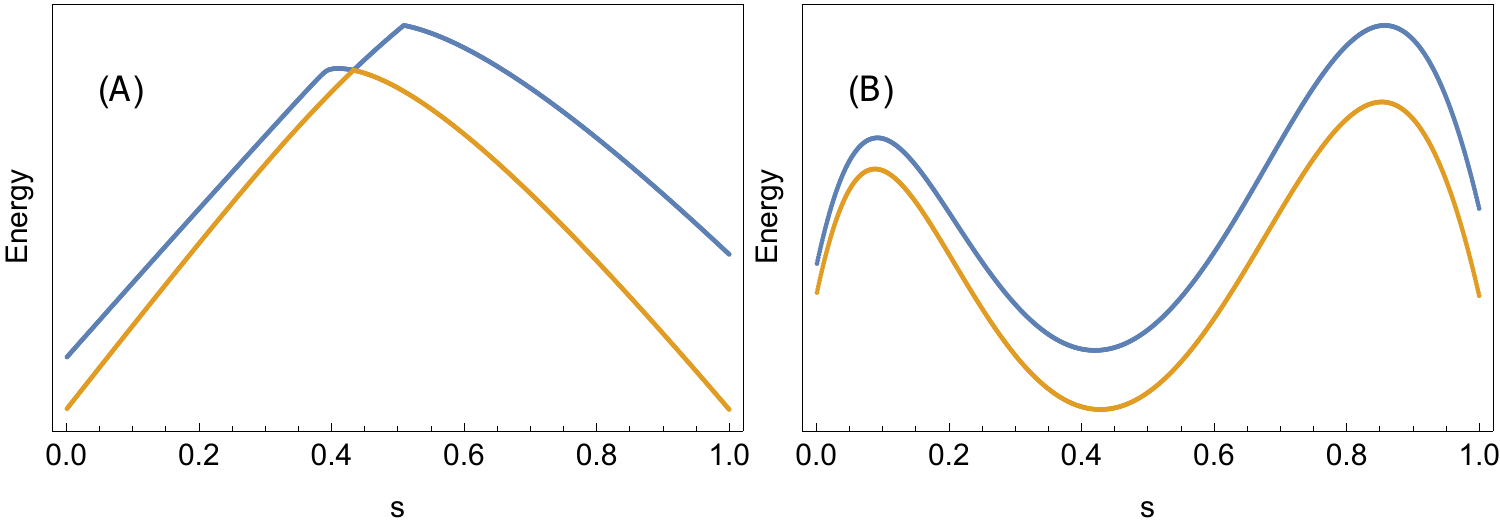}
\end{center}
\caption{\textbf{(A)} The ground state energy and first excited state energy of the Hamiltonian $H(s)$ without the path change.  Note the avoided crossing at $s = 0.434...$ which corresponds to an exponentially small gap. \textbf{(B)} The ground state energy and the first excited state energy of the Hamiltonian $H(s) + s(1-s)H_E$ that includes the path change term.  The spectral gap of the system does not go to zero with the system size anywhere in the interval $0 \leq s \leq 1$.}\label{fig:spectraPathChange}
\end{figure}
\begin{figure}[h]
\begin{center}
\includegraphics[scale=1,trim=20 0 0 0]{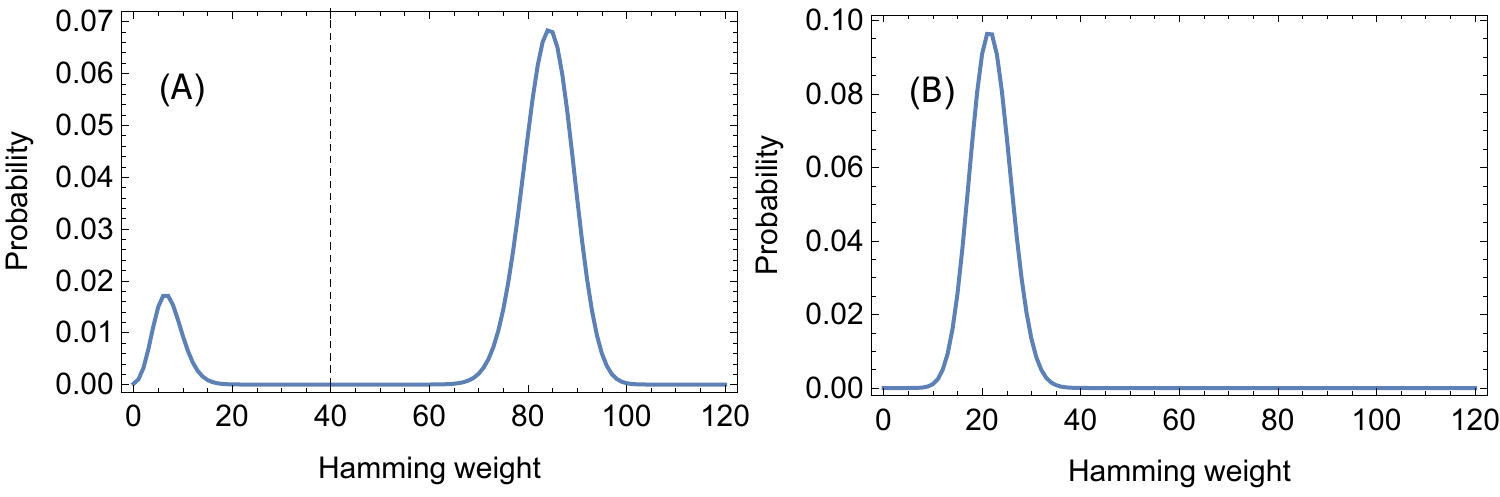}
\end{center}
\caption{\textbf{(A)} The ground state probability distribution of the Hamiltonian $H(s)$ without the path change at the point $s = 0.434...$ where the minimum spectral gap occurs.  The dashed vertical line separates the distribution into two regions with exponentially small vertex expansion, corresponding to the exponentially small spectral gap of the system.  \textbf{(B)} The ground state probability distribution of the Hamiltonian $H(s) + s(1-s)H_E$ that includes the path change term at the point $s = 0.24...$ where the spectral gap is minimized.  There is no way to separate this ground state into regions of low expansion, which corresponds to the fact that the spectral gap does not go to zero with the system size.}\label{fig:gsPathChange}
\end{figure}
\section{Outlook}
The bound \eqref{eq:expansion} can be viewed as a fundamental limitation on the kinds of quantum states that can be efficiently produced using the strict model of adiabatic computation in which the system is idealized to always remain in the ground state of a local Hamiltonian and the run time scales polynomially with the inverse of the minimal spectral gap along the evolution.  In the context of this idealized model, we have shown that adiabatic optimization with local Hamiltonians is efficient only if there are never any small bottlenecks in the ground state probability distribution along the adiabatic path.  Therefore efforts to improve the performance of idealized adiabatic computation using ``path changes'' should focus on how these terms can alter the ground state probability distribution to remove or reduce the effect of geometric bottlenecks.  Towards this effort we note that lack of ground state bottlenecks in the computational basis guarantees a large spectral gap for irreducible stoquastic Hamiltonians such as the transverse Ising model by Theorem \ref{thm:stoquastic}, but such a guarantee does not necessarily hold when a non-stoquastic Hamiltonian is considered in an arbitrary basis according to the counterexample of section \ref{sec:counterexample}.    

Much of the attention on the idealized model of adiabatic computation described above is due to the fact that the adiabatic theorem provides rigorously sufficient conditions for the run time of the algorithm.  These conditions are not necessary however, and a number of works have pointed out that phenomena like diabatic transitions to excited states~\cite{crosson2014different,muthukrishnan2016tunneling} and thermal relaxation~\cite{smelyanskiy2015quantum} can improve the performance of quantum annealing as contrasted with purely adiabatic computation.  Furthermore, the example of universal adiabatic computation shows that measurements can be used to efficiently project the state of an adiabatic computer into a state with strongly localized modes such as a GHZ state.  In the case of universal AQC this involves a clever design of the Hamiltonian, and it would be interesting to investigate how intermediate measurements could improve adiabatic computation more generally.  While non-stoquastic driver Hamiltonians remain an exciting future direction for improving the performance of  adiabatic optimization, our results suggest that future quantum annealers may find an even greater benefit by taking advantage of excited states and non-adiabatic effects.  

\paragraph{Acknowledgements.}
E.\,C. is grateful for support provided by the Institute for Quantum Information and Matter, an NSF Physics Frontiers Center (NSF Grant PHY-1125565) with support of the Gordon and Betty Moore Foundation (GBMF-12500028).  J.B. is grateful for support provided by the Caltech Summer Undergraduate Research Fellowship program, and also thanks the IQIM for hospitality.  We thank T. Cubitt for illuminating discussions that inspired the application of our results to history-state Hamiltonians.

\bibliography{isoperimetric}
\bibliographystyle{abbrv}

\end{document}